\documentclass[runningheads]{llncs}
\usepackage[T1]{fontenc}
\usepackage{multirow}

\usepackage{graphicx}
\usepackage{ifthen}

\usepackage{amssymb}
\usepackage{amsmath}

\usepackage{bm}

\usepackage{chngcntr}
\usepackage{enumitem}

\usepackage{bigints}

\usepackage{color}
\usepackage[usenames,dvipsnames,svgnames,table]{xcolor}

\usepackage{bbm}

\usepackage{hyperref}
\usepackage{ulem}

\newcommand{\floor}[1]{\left\lfloor #1 \right\rfloor}

\newcolumntype{C}[1]{>{\centering\arraybackslash$}p{#1}<{$}}

\newcommand{\mket}[1]{| #1 \rangle}

\newcommand{\mtr}[1]{
	\ifthenelse{ \equal{#1}{} }
	{ \mathrm{Tr} }
	{ \mathrm{Tr}\left( #1 \right) }
}

\newcommand{\mptr}[2]{
	\ifthenelse{ \equal{#1}{} }
	{	\mathrm{Tr_{#2}} }
	{	\mathrm{Tr_{#2}}\left( #1 \right)}
}

\newcommand{\ff}{\mathfrak{f}}

\newcommand{\detD}{\mathfrak{D}}
\newcommand{\detDr}{\mathfrak{D}_r}

\newcommand{\EH}{ E(\nH) }

\newcommand{\vN}{ \mathrm{vN} }
\newcommand{\vNH}{ \mathrm{vNH} }

\newcommand{\QinEH}{Q \in E(\nH)}

\newcommand{\bI}{\mathbb{I}}

\newcommand{\bC}{\mathbb{C}}

\newcommand{\bR}{\mathbb{R}}
\newcommand{\bN}{\mathbb{N}}

\newcommand{\nB}{\mathcal{B}}

\newcommand{\nE}{\mathcal{E}}
\newcommand{\nH}{\mathcal{H}}

\newcommand{\nL}{\mathcal{L}}

\newcommand{\nU}{\mathcal{U}}

\newcommand{\nP}{\mathcal{P}}

\begin{document}
\title{The Infinite-Dimensional Quantum Entropy:\\the Unified Entropy Case}
%
%
\author{Roman Gielerak\inst{1}\orcidID{0000-0001-8657-0829} \and
Joanna Wi\'sniewska\inst{2}\orcidID{0000-0002-2119-3329} \and
Marek Sawerwain\inst{1}\orcidID{0000-0001-8468-2456}}
\authorrunning{R. Gielerak et al.}
%
\institute{Institute of Control \& Computation Engineering, University of Zielona G\'ora, Licealna 9, Zielona G\'ora 65-417, Poland\\ \email{\{R.Gielerak, M.Sawerwain\}@issi.uz.zgora.pl} \and
Institute of Information Systems, Faculty of Cybernetics,\\ Military University of Technology,  Gen.~S.~Kaliskiego 2, Warszaw 00-908, Poland\\
\email{JWisniewska@wat.edu.pl}\\
}
\maketitle              
\begin{abstract}
By a use of the Fredholm determinant theory, the unified quantum entropy notion has been extended to a case of infinite-dimensional systems. Some of the known (in the finite-dimensional case) basic properties of the introduced unified entropies have been extended to the case study. Certain numerical approaches for computing the proposed finite and infinite-dimensional entropies are being outlined as well.
\end{abstract}
\section{Introduction} \label{lbl:sec:introduction}

The entropy 
is one of the most significant tools in the information theory, both in the classical and quantum approach \cite{Nielsen2001}. To simplify a bit, the entropy, in the quantum context which is considered in this work, describes the level of randomness applied as a quantitative measure of entanglement \cite{Bengtsson9999} in an analysed quantum state which is broadly utilised in different fields of quantum computations. 

The von Neumann entropy is the most popular notion, so it is also termed as quantum entropy. There are several important extensions of the entropy notions, in particular the conditional entropy and relative  entropy which play a crucial role in the quantum information theory.  In addition, we may also apply modifications of the original von  Neumann entropy notions, such as: Quantum-MIN entropy, Tsallis entropy, R\'enyi entropy, and unified entropies (which are discussed in this work). It is a basic fact that in the infinite dimensional case the introduced quantum entropies notions are, in general, not continuous on the full spaces of quantum states. It is very important that the entropy value may be finite or infinite. If we describe quantum states as finite density matrices and utilise the most (and many other as well) popular computational model -- the quantum circuits -- the value of entropy is finite. However, a continuous-variable quantum computation model is also considered as very important and highly usable computational model where a system's states are represented in the infinite Hilbert spaces. In this case, the entropy value might be also infinite. In the present paper, we discuss a renormalised variant of unified entropy which allows computing finite entropy values for states given by the infinite dimensional Hilbert space with the use of the Fredholm determinants \cite{Fredholm1909} technique.

The approach presented in this work enables us to calculate the entropy values by the standard linear algebra packages -- we show an exemplary numerical computations carried out with the use of Python code supported by the libraries NumPy and EntDetector \cite{EntDetector2021} which is dedicated to calculating the level of quantum entanglement. The EntDetector package also contains functions computing entropy values by the technique of Fredholm determinants, both for finite and infinite cases.

The paper is organised as follows: in Sec.~\ref{lbl:subsec:notation} we present foundations concerning the utilised notation. The entropy notion for bipartite systems is defined in Sec.~\ref{lbl:sec:entropy:bipartite}. The method of calculating and the renormalisation of von Neumann entropy with the use of Fredholm determinants technique is outlined Sec.~\ref{lbl:sec:q:Neumann:entropy}. The notion of the unified entropy for the finite and infinite cases is introduced in Sec.~\ref{lbl:sec:unified:entropy}. In Sec.~\ref{lbl:numerical:examples}, we present some numerical examples. Conclusions are contained in Sec.~\ref{lbl:conclusions}. Acknowledgments and References end the paper.    

\subsection{On the standard notation used} \label{lbl:subsec:notation}

Sets of real, complex, and integer numbers will be denoted as $\bR$, $\bC$, and $\bN$, respectively. Small letters as $d$, $i$, $j$, $k$, $l$, $n$ always stand for integer numbers and are used to denote indexes, dimensionality or cardinal/ordinal numbers. By the letter $\nH$, we always denote a separable Hilbert space, $\dim( \nH ) \leq \infty$ stands for its dimension. The $C^{\star}$-algebra of bounded linear operators acting in $\nH$ will be denoted as $\nB(\nH)$, and the operator norm as $|| \cdot ||$. The set of all states on $\nH$ will be marked as $E(\nH)$ and its boundary consisting of pure states is denoted as $\partial E(\nH)$. The multiplicative group of unitary maps acting in $\nH$ is denoted as $\nU(\nH)$. By $\nB_{+}( \nH )$ we denote the set of bounded and positive operators acting on space $\nH$: $\nB_{+}( \nH ) = \{ Q \in \nB(\nH) : Q \geq 0\}$.

$\mtr{ \cdot }$ stands for the standard trace map defined  in the trace-class compact operators acting in $\nH$. The corresponding operator's ideals equipped with the standard Schatten norm $\nL_p$, $p \geq 1$, will be denoted as $\nL_p( \nH )$. For $Q \in \nL_p( \nH )$ the spectrum of $Q$ will be always denoted as $\sigma(Q) = (\lambda_n)$, $n=1,\ldots,\infty$ and sorted in such way that $\lambda_i \geq \lambda_{i+1}$.

If $\nH = \nH^A \otimes \nH^B$, then the corresponding partial trace taking quantum operations will  be denoted as $\mptr{ \cdot }{ A }$, resp. $\mptr{ \cdot }{ B }$. In particular, if $Q \in E( \nH )$ then the corresponding reduced density matrices will be denoted as $Q_B$ and resp. $Q_A$.

\subsection{Entropy based entanglement measures in bipartite systems} \label{lbl:sec:entropy:bipartite}

Let us consider a bipartite system "A+B", the Hilbert space of states which is given as $\nH = \nH^{A} \otimes \nH^{B}$, with $\dim \nH^A \cdot \dim \nH^B \leq \infty$.

Being motivated mainly by the work \cite{Plenio}, a map $\nE$:
$\nE : E( \nH ) \longrightarrow [0, \infty]$,
will be called an $H$-entropy-based entanglement measure iff obeys the following properties:
\begin{itemize}

	\item[en(1):] for $Q \in \partial E(\nH)$:
		$\nE( Q ) = ( H \circ \mptr{}{B} )(Q) = ( H \circ \mptr{}{A} )(Q)$,

	\item[en(2):] if $\QinEH$ is separable, then
		$\nE( Q ) = 0$,
	
	\item[en(3):] $\nE$ is non-increasing under local quantum operations,
	
	\item[en(4):] the measure $\nE$ should be invariant under the action of local unitary groups.
		
\end{itemize}
The introduced in en(1) map $H : E(\nH) \longrightarrow [0, \infty]$ is called entropy-like map iff:
\begin{itemize}

	\item[ent(1):] the map $H$ is concave (or convex) and finite (or continuous in $\nL_1(\nH)$-norm) on $E(\nH)$,

	\item[ent(2):] $\forall_{ Q \in \partial E(\nH)} H(Q) = 0$,
	
	\item[ent(3):] $\forall_{U \in \nU(\nH)} H( U Q U^{\dagger} ) = H(Q)$,
	
	\item[ent(4):] $H$ is non-increasing under the action of quantum operations,

	\item[ent(5):] if $\nH = \nH^A \otimes \nH^B$,  $Q \in E(\nH)$ then $| H(Q_A) - H(Q_B)| \leq H(Q)$.
	
\end{itemize}

The basic, common elements building the class of entropies discussed in the present note, follow a map $I_r$ for $r \in (0, 1) \cup (1, \infty)$ (the case of von Neumann entropy corresponding to the choice $r=1$ is very briefly discussed below in Section~\ref{lbl:sec:q:Neumann:entropy}):
\begin{equation}
		I_r : E(\nH) \rightarrow (0, \infty], \;\;\; I_r(Q)  =  {|| Q^r ||}_1 = \sum_{ \lambda \in \sigma(Q)} \lambda^r .
\label{lbl:eq:def:funf}
\end{equation}
For $r \geq 1$, the map $I_{ r }$ is exactly the Schatten class operator norm \cite{Simon1977} and it is widely used in several applications of the ideals in the operator algebras, see i.e. \cite{Simon2005}.
Assuming $\dim(\nH) = \infty$ and $r \in (0, 1)$,  the situation with $I_r$ definition is much more complicated. In fact, the following proposition is valid.

\begin{proposition}
Let $\nH$ be a separable Hilbert space with $\dim(\nH) = \infty$ and let $r \in (0, 1)$. Then the set $I^{\infty}_{r}(\nH) = \{ Q \in E( \nH ) : I_r(Q) = \infty \}$ is $L_1$ -- dense subset of $E(Q)$.
\end{proposition}

\begin{proof}
Let $Q^\epsilon$ be a state in $E(\nH)$, $\epsilon > 0$ with the following spectrum $\sigma(Q^\epsilon) = (z^{-1}_{\epsilon} \cdot \frac{1}{k^{1+\epsilon}})_k$, $z_{\epsilon} = \sum_{k=1}^{\infty} \frac{1}{k^{1+\epsilon}} < \infty$. Then, for $s \in (0, 1)$ obeying $s \leq \frac{1}{1+\epsilon}$, $I_s(Q^\epsilon) = \infty$.

Let us choose $\QinEH$ with spectrum $\sigma( Q ) = {(\lambda_k)}_{k}$. For an arbitrary small $\delta > 0$ and arbitrary large $M > 0$ there exists a number $K(\delta, M)$ such that
\begin{equation}
	\frac{1}{z_{\epsilon}} \cdot \sum_{k \geq K(\delta, M)} \frac{1}{ k^{1+\epsilon} } < \delta ,
\end{equation}
and, for $s < \frac{1}{1+\epsilon}$:
\begin{equation}
	\sum_{k \geq K(\delta, M)} \frac{1}{ k^{s(1+\epsilon)} } \geq z^s_\epsilon \cdot M .
\end{equation}

Now, we form the following spectral set:
\begin{equation}
{(\sigma_{\delta, M})}_{k} = z^{-1} \left\{ \begin{array}{l}
\frac{1}{z_{\epsilon}} \lambda_k , \; \mathrm{for} \; k \leq K(\delta, M), \\
\\
\frac{1}{k^{1+\epsilon}} , \; \mathrm{for} \; k > K(\delta, M) ,
\end{array}
\right.
\end{equation}
where $z_{\epsilon} = \sum_{k} ( \sigma_{\delta, M} )_{k} < \infty$ uniformly in $\delta$ and $M$.

Let $Q \otimes_{\delta, M} Q^{\epsilon}$ be any state  with the spectrum equal to $\sigma_{\delta, M}$. Then 
	$|| Q - Q \otimes_{\delta, M} Q^{\epsilon} ||_1 < \delta$ ,
and $s < \frac{1}{1+\epsilon}$:
$I_s ( Q \otimes_{\delta, M} Q^{\epsilon} ) = \infty$ .
\qed
\end{proof}

It is the main motivation for the present note to propose how to overcome this severe problem that we meet in the case of infinite dimensional systems.

For this goal, the theory of a regularised Fredholm's determinants has been proposed \cite{GielerakSpringer}, \cite{GielerakSawerwainInPreparation} and briefly outlined in the case of the standard von Neumann entropy and some two-parameters deformations known under the name: unified entropy of Hu and Ye \cite{Hu2006}. In the class of entropies analysed in the present note, the well known examples of the  one-parameter deformations of the von Neumann entropy, widely known as Tsallis and R\'enyi entropies entropies  are included.

\section{Quantum von Neumann Entropy and Fredholm Determinants} \label{lbl:sec:q:Neumann:entropy}

In this subsection, we outline some of the recent results that we have obtained with the use of the Fredholm determinants theory in \cite{GielerakSpringer}, \cite{GielerakSawerwainInPreparation}.

Let $\nH$ be a separable Hilbert space and let $\dim(\nH) = \infty$. For $\QinEH$, it was proved in \cite{GielerakSpringer}, \cite{GielerakSawerwainInPreparation} that the subset
$\vN( \infty ) = \{ \QinEH : \mtr{-Q \log Q} = \infty \}$,
is dense subset (in $L_1$ - topology) in $E(\nH)$.

For $Q \in vN( \infty )^c$, i.e. in  the case $\mtr{-Q \log Q} < \infty$ which is equivalent to $(Q^{-Q} - \bbbone) \in \nL_1(\nH)$, it is possible to prove that then,  the following Fredholm determinant \cite{GielerakSawerwainInPreparation}:
\begin{equation}
	\detD(Q) = \det( \bbbone_{\nH} + \ff(Q)) ,
	\label{lbl:eq:detD:def}
\end{equation}
where $\ff(Q) = Q^{-Q} - \bbbone_{\nH}$ is finite and moreover
\begin{equation}
	\vNH(Q) = - \mtr{Q \log Q} = \log \detD(Q) . 
	\label{lbl:eq:vNH}
\end{equation}

Using the technique developed for the analysis of the Fredholm determinants \cite{Simon1977}, \cite{Simon2005}, it was shown in \cite{GielerakSawerwainInPreparation} that all the basic facts known from the quantum von Neumann entropy (\ref{lbl:eq:detD:def}) in the finite dimensional case such as: certain type of continuity, together with ent(1) -- ent(5) besides many other do hold,  could be extended to the subset $vN(\infty)^c$ of $E(\nH)$, where $\vN(\infty)^c = E(\nH) \setminus \vN(\infty)$.

Moreover, the following fact has been proved in \cite{GielerakSawerwainInPreparation}:
$\mathrm{if}  \; \QinEH \; \mathrm{then} \; (Q^{-Q} - \bbbone_{\nH} ) \in \nL_2(\nH)$ .

This enables us to write down the following, renormalised version of the von Neumann entropy:
\begin{equation}
\vNH^{ren}(Q) = \mtr{ -Q \log Q + (Q^{-Q} - \bbbone) },
\label{lbl:eq:vNH:ren}
\end{equation}
that appear to be finite and continuous (in the $\nL_2(\nH)$ topology)  on the whole set of quantum states $E( \nH )$. The proof is obtained with the use of the regularised Hilbert-Fredholm determinants techniques, see \cite{GielerakSpringer}, \cite{GielerakSawerwainInPreparation}.

\begin{example}
Let $Q \in E(\nH)$ be such that $\sigma(Q) = ( \lambda_n )_{n} = \frac{1}{z_{\beta}} ( \frac{1}{ n \log \beta_n} )_{n}$, for $\beta \in (1, \infty)$, $z_\beta = \sum_{n=1}^{\infty} \frac{1}{n \log^{\beta} n } < \infty$. It is easy to check that $- \sum \lambda_n \log \lambda_n = \infty$ for $\beta \in (1, 2)$. However the renormalised entropy:
\begin{multline}
	H^{\mathrm{ren}} (Q) = \mtr{ -Q \log Q + (Q^{-Q} - \bbbone ) } = \\ \log \det ( \bbbone_{\nH} + (Q^{-Q} + \bbbone) ) e^{ - \mtr{ Q^{-Q} - \bbbone }} < \infty .
\end{multline}
\end{example}

\section{The unified quantum entropies in terms of the Fredholm determinants} \label{lbl:sec:unified:entropy}

\subsection{The Hu-Ye Unified Entropy for the Finite-dimensional Case}

Let us recall the notion, together with some basic properties, of the two-parameter deformation of the von Neumann entropy as given in \cite{Hu2006} by Xinhua Hu and Zhongxing Ye. Let $ d = \dim( \nH )< \infty $. For $r \in (0,1) \cup (1,\infty)$ and $s \in \mathbb{R} \setminus \{0\}$ the Hu-Ye unified entropy $HY_r^s$ is defined as:
	$HY_r^s : E( \nH ) \to [0,\infty)$,
and
$HY_{r}^{s}(Q) = \frac{1}{(1-r)s} \left( (\mtr{ Q^{r} }^{s} - 1 \right)$.
Some of the basic, albeit selected, properties of this quantum entropy version are collected in the next subsection~\ref{lbl:sec:hy:entropy:summary}.

\subsection{HY-entropy summary ($d < \infty $)} \label{lbl:sec:hy:entropy:summary}

We recall following basic properties of the HY-entropy.
\begin{itemize}
\item[(HY1)] Connection with other entropies:
\begin{itemize}
	\item[(i)]  
		$\lim_{s \to 1} HY_{r}^{s}(Q) = HY_{r}^{1}(Q) = T_r(Q)$,
		where $T_r$ stands for the Tsallis entropy functional and the limit is taken pointwise on $E(\nH)$.
	
	\item[(ii)] for any admissible value of $r$ and $s=1$:
		$\lim_{r \to 1} HY_{r}^{s}(Q) = H(Q)$,
	
	\item[(iii)] and for any admissible value of $r$:
		$\lim_{s \to 0} HY_{r}^{s}(Q) = R_r(Q)$,
		where $R_r(Q)=\frac{1}{1-r} \log(\mtr{Q^r})$ is the R\'enyi entropy.
\end{itemize}

\item[(HY2)] Non-negativity and boundness for any admissible values of $r$ and $s$:
$\forall_{ \QinEH } \; \; 0 \leq HY_r^s(Q)  \leq \frac{1}{(1-r)s} (d^{(1-r)s} - 1)$
and:
	\begin{itemize}
	\item[(i)]
		$HY_r^s(Q)=0, \; \; \textrm{iff} \; \; Q \in \partial E(\nH)$, 
	if $\mtr{Q^2} = 1$, see \cite{Hu2006},
	
	\item[(ii)]
		$HY_{r}^{s}(Q) \leq \frac{1}{(1-r)s} ({d'}^{(1-r)s} - 1)$,
	iff $\sigma(Q)=(\frac{1}{d'}, \ldots, \frac{1}{d'}, 0, \ldots, 0)$ where ${d'} = \mathrm{rank}(Q)$ .
	\end{itemize}

\item[(HY3)] If $\nH = \nH_{A} \otimes \nH_{B}$ then for any $\QinEH$:
\begin{itemize}
\item[(i)] 
$\forall_{U \in \nU(\nH)} \; \; HY_r^s(U^\dagger QU) = HY_r^s(Q)$,

\item[(ii)] let:
$Q_A = \mptr{Q}{B} , \; \; \; Q_B = \mptr{Q}{A}$,
then
$HY_{r}^{s}(Q_A) = HY_{r}^{s}(U_A^\dagger Q_A U_A)$ and $HY_{r}^{s}(Q_B) = HY_{r}^{s}(U_B^\dagger Q_B U_B)$,
where $U_{A(B)} \in \nU(\nH_{A(B)})$.

\item[(iii)]  if $Q \in \partial E(\nH_A \otimes \nH_B)$ then for any admissible $(r,s)$:
$HY_{r}^{s}(Q_A) = HY_{r}^{s}(Q_B)$.

\end{itemize}

\item[(HY4)] Continuity. It is known \cite{Hu2006} that for $r>1$ and $s \geq 1$:
$|HY_r^s(Q) - HY_r^s(Q')| \leq \frac{1}{r(r-1)} {||Q-Q'||}_1$.

\item[(HY5)] Concavity. Let $r \in (0,1)$, $s > 0$, and $r \cdot s < 1$ or $r \geq 1$, $r \cdot s \geq 1$, and let $Q=\sum_{k} \lambda_k Q_k$, $Q_k \in E( \nH )$, $\lambda_k \in [0,1]: \sum_{k} \lambda_k = 1$. Then (see \cite{Hu2006}):
$HY_r^s(Q) \geq \sum_{k} \lambda_k HY_r^s(Q_k)$.

\item[(HY6)] Triangle inequality. Let $\nH = \nH_A \otimes \nH_B$ and $Q \in E(\nH_A \otimes \nH_B)$, then for $r>1$ and $s \geq r^{-1}$ (see \cite{Rastegin2011}):
\begin{equation}
|HY_r^s(Q_A) - HY_r^s(Q_B)| \leq HY_r^s(Q) .
\label{lbl:eq:HYsr:triangle:inequality}
\end{equation}
\end{itemize}

\subsection{The Hu-Ye entropy in infinite dimensional case}

It will be assumed in the present subsection that $\dim(\nH) = \infty$ and the Hilbert space $\nH$ is separable. Let, for $r > 0$,
\begin{equation}
	\ff_r : \EH \longrightarrow \nB_{+}( \nH ), \;\;\; 	Q \longrightarrow e^{Q^r} - \bI .
\end{equation}

\begin{lemma} ~ 

	\begin{itemize}
		
		\item[(1)] For any $\QinEH$, $r \geq 1$, 
			\begin{equation}
				1 \leq {||  \ff_r(Q) ||}_{1} \leq e .
				\label{lbl:lemma:eq:ineq:1}
			\end{equation}
			
		\item[(2)] For $r \in (0, 1)$,
				${|| Q^{\alpha} ||}^{r\alpha}_{r\alpha} \leq {|| \ff_r(Q) ||}_{\alpha}^{\alpha} \leq e^{\alpha} {|| Q ||}_{r\alpha}^{r\alpha}$
		for any $\alpha \geq \frac{1}{r}$.
	
	\end{itemize}
\end{lemma}
\begin{proof}
From the elementary estimate for $\lambda \in [0,1]$:
\begin{equation}
	\lambda \leq e^{\lambda} - 1 = \int_{0}^{1} e^{\tau \lambda} \lambda \; d \tau \leq e \cdot \lambda,
\end{equation}
it follows
\begin{equation}
	\sum_{ \lambda \in \sigma(Q)} \lambda^r \leq \sum_{ \lambda \in \sigma(Q)} (e^{\lambda^r} - 1) \leq e \cdot \sum_{ \lambda \in \sigma(Q)} \lambda^r .
\end{equation}

Assuming $r \geq 1$, the inequality~Eq.~\ref{lbl:lemma:eq:ineq:1} follows. In the case $r \in (0, 1)$ and for any $\alpha$ such that $\alpha r > 1$: 
\begin{equation}
	\sum_{ \lambda \in \sigma(Q)} \lambda^{r\alpha} \leq {||  \ff_r(Q) ||}_{\alpha}^{\alpha} = \sum_{ \lambda \in \sigma(Q)} (e^{\lambda^r} - 1)^{\alpha} \leq e^{\alpha} \cdot \sum_{ \lambda \in \sigma(Q)} \lambda^{\alpha r}.
\end{equation}
\qed
\end{proof}

As a corollary of this lemma, we have the succeeding proposition.

\begin{proposition}{~}
Let $\QinEH$,
\begin{itemize}
	
	\item[(1)] if $r > 1$ then the Fredholm determinant defined as
		 \begin{equation}
		 \begin{array}{lcl}
		 	\detDr : E(Q)  & \longrightarrow & \detDr(Q) = \det (\bbbone + \ff_r(Q)) ,
		 \end{array}
		 \end{equation}
	 exists (is finite) and obeys the following properties
		\begin{itemize}
			
			\item[(i)] $1 \leq \detDr(Q) < e^{e}$,
			
			\item[(ii)] if ${|| \ff_r(Q_n) - \ff_r(Q)||}_{1} \rightarrow 0$ as $n \to \infty$ then $\detDr(Q_n) \longrightarrow \detDr(Q)$,
			
			\item[(iii)] the succeeding equalities are valid
				\begin{equation}
					\detDr(Q) = \exp ( \mtr{\log( \bbbone + \ff_r(Q) ) } ) = \sum_{n=0}^{\infty} \mptr{ \Lambda^n(Q) }{\Lambda^n(\nH)} .
				\end{equation}		
		\end{itemize}	 
	\item[(2)]  Let $r \in (0, 1)$. Then, for any $\alpha > 0$ such that $\alpha r > 1$ the following renormalised Hilbert-Fredholm determinant is finite:
	\begin{equation}
		\detD^{ren}_{r, \alpha}(Q) = {\det}_{\alpha} (\bbbone + \ff_r(Q)) =  \detDr(Q)  \exp{ \left( \sum_{j=1}^{\alpha-1} (-1)^{j+1} \frac{\mtr{\ff_r(Q)}^j}{j} \right)} .
		\label{lbl:eq:detDren:def:r:alpha}
	\end{equation}
	
\end{itemize}
\end{proposition}

With the use of the introduced determinants, we can rewrite the unified entropy formula of Hu-Ye:
\begin{itemize}

\item[(I)] if $r > 1$, $s \neq 0$
\begin{equation}
\forall_{\QinEH} \; \mathrm{HY}^{s}_{r}(Q) = \frac{1}{(r-1)s} \left( {( \log( \detD_r(Q) ) )}^s - 1 \right) ,
\label{lbl:eq:HY:v1}
\end{equation}

\item[(II)] if $r \in (0, 1)$, $s \neq 0$ the renormalised Hu-Ye entropy formula is given as 
\begin{equation}
\forall_{\QinEH} \; \mathrm{renHY}^{s}_{r}(Q) =  \frac{1}{(r-1)s} \left( {( \log (\detD^{ren}_{r,\alpha}(Q)) )}^s - 1 \right),
\label{lbl:eq:HY:v2}
\end{equation}
for $\alpha = \min \{ \alpha \in N : \alpha r \geq 1 \}$.
\end{itemize}

The main results of our investigations are being formulated now:

\begin{theorem}
Let $\dim( \nH ) = \infty$ and $\nH$ is separable. Let r > 1 then, for suitable value of $s \neq 0$ all the properties listed as HY(1), HY(2) (with the form independent of d), HY(3), HY(4), HY(5) and HY(6) are true, for $\mathrm{HY}^s_r$ on $E(Q)$.
\end{theorem}
\begin{proof}
The main argument used for the proof 
is the proper control (in this case, $r>1$) on the finite dimensional approximations. \qed
\end{proof}

\begin{theorem}
Let $\dim( \nH ) = \infty$ and let $\nH$ be separable. Let $r \in (0, 1)$, $s \neq 0$. Then, for any $Q \in E(Q)$ there exists $\alpha^{\star} \in N$ such that $\alpha^{\star} = \min \{ \alpha \in N: (e^{Q^{\alpha \cdot r}} - \bbbone) \in \nL_1(\nH) \}$ and the renormalised HY-entropy as given in (\ref{lbl:eq:detDren:def:r:alpha}) is finite and continuos in $\nL_{\alpha_{\star}}(\nH)$ -- norm.
\end{theorem}

\begin{remark}
The problem whether and which one (from the list HY(1) -- HY(5) as listed for the case $\dim(\nH) < \infty$) and other properties like strong-subadditivity 
\cite{Lieb1973b} are preserved  under the proposed here renormalisation approach are under investigations. 
\label{lbl:remark:HY:properties:question}
\end{remark}

\begin{example}
Let $\zeta$ be the zeta function of Riemann 
i.e.
\begin{equation}
\zeta(q) = \sum_{n=1}^{\infty} \frac{1}{n^q}, \; \mathrm{for} \; q > 1.
\end{equation}
Let $\nP$ be the set of primes. The operator $Q^{\zeta}_{ q, r } \in E(\nH)$ is such that
\begin{equation}
	\sigma( Q^{\zeta}_{q,r} ) = \bigg( z^{-1} ( \log( 1 + \frac{1}{p^q_k}) )^{1/r} \bigg)_{k}, \;\;\; \mathrm{where} \;\;\; 
	z = \sum_{k=1}^{\infty}  \bigg( \log( 1 + \frac{1}{p^q_k}) \bigg)^{1/r} ,
\end{equation}
which is finite for $s > 1$, $r > 1$ and $p_k \in \nP$ is the k-th prime (numbered as k-th in the natural ordering of $\nP$). Then, by the use of formula:
$\prod_{p \in \nP}  \bigg( 1 + \frac{1}{ p^q }  \bigg) = \frac{ \zeta( q ) }{ \zeta( 2q ) }$,
it follows that
$\detD_r( Q^{\zeta}_{ q, r } ) =  \frac{ \zeta( q ) }{ \zeta( 2q ) }$
and therefore, for $s=1$, 
\begin{equation}
	 \frac{ \zeta( q ) }{ \zeta( 2q ) } = \exp \bigg(  {(r-1)}   H^1_r  ( Q^{\zeta}_{ q,r } ) +1   \bigg).
\end{equation}
Defining the corresponding to $Q^{\zeta}_{s,r}$ Hamiltonian:
 $h^{\zeta}_{ q,r } = \log Q^{\zeta}_{ q, r } \geq 0$ ,
which is self-adjoint and positive, we enrich the class of Hamiltonians the spectrum of which are connected to the zeta function of Riemann 
\cite{Berry1999}, \cite{Bender2017}.
\end{example}

\section{Numerical Examples} \label{lbl:numerical:examples}

In this section, we present a two numerical examples related to quantum states for which we calculate entropy by the use of the Fredholm determinant technique. A source code in Python language of all discussed below examples is publicly available at \cite{RepoSourceCodeGWJ}. The first example is devoted the quantum state X and unified renormalised entropy like in Eq.~\ref{lbl:eq:HY:v1} and Eq.~\ref{lbl:eq:HY:v2} where for a given dimensionality, we check the correctness of the triangle inequality (HY6) with the entropy expressed by Eq.~\ref{lbl:eq:HY:v1}. In other words, we calculate partial traces for subsystem A~and B~and compare subsystem's entropies to the entropy of the whole quantum state X. In the second example, we study numerical calculations of entropy value for the two–mode squeezed Gaussian states.

Before the presentation of examples, let us notice that the entropy's computational complexity calculated directly, i.e. by Eq.~\ref{lbl:eq:HY:v2}, depends on a complexity of the function pointing out logarithm of a matrix and the operator $\ff_r( \cdot )$. Therefore, if $T$ stands for a general computational complexity, where $n$ marks a dimension of the operator $Q$, we obtain:
\begin{multline}
T(n) = T_s(1) + T_D(Q) = T_s(1) + \bigg( T_D(Q) + \sum_{j=1}^{\alpha} T_{\mathrm{Tr}}(Q) + T_{\ff}(Q) \bigg) \\ = O(1) + \bigg( O(n^3) + \sum_{j=1}^{\alpha} O(n^2) + O(n^3) \bigg) = O(n^3),
\end{multline}
where $T_s$ means constant values, e.g. a fraction, $T_D$ is the complexity of determinant $\detD^{ren}$, $T_{\mathrm{Tr}}$ determines the time of calculating a trace, and $T_{\ff}$ is a complexity of the function $\ff_r$. Generally, the matrix functions may be computed with the use of the spectral decomposition with its complexity $O(n^3)$ what determines a complexity of the renormalised Hu-Ye entropy (Eq.~\ref{lbl:eq:HY:v2}).  

In further part of this section, we also present an entropy approximation for the Gaussian bipartite state by a procedure according \cite{Bornemann2010} which is: 

\begin{verbatim}
def fredholm_det(K, z, a, b, m):
    w,x=gauss_legendre_quadrature(a,b,m)
    w = np.sqrt(w) 
    xi,xj = np.meshgrid(x, x, indexing='ij')
    d = np.linalg.det( np.eye(m) + z * np.outer(w,w) * K(xi,xj) )
    return d
\end{verbatim}
This procedure's complexity may be described as:
\begin{equation}
T(n) = T_{QLQ}(m) + T_s(n) + T_{det}(Q) = O(m^2) + O(n^2) + O(n^3) = O(n^3).
\end{equation}
The most dominant operation here is calculating the determinant ($T_{det}$) because the kernel function $K$ is computed linearly (depending on the number of points in a given quadrature). While, calculating the quadrature's coefficients $T_{QLQ}(m)$ is characterised by the quadratic complexity depending on the parameter $m$.

\subsection{The $d$ dimensional X quantum state}

We also examine one of possible generalisations of the X-type quantum state which can be depicted in two form depending on whether the $n$ is even $Q^{Xe}_{d}$ or odd $Q^{Xo}_{d}$:

\begin{equation}
	\includegraphics[width=0.95\textwidth]{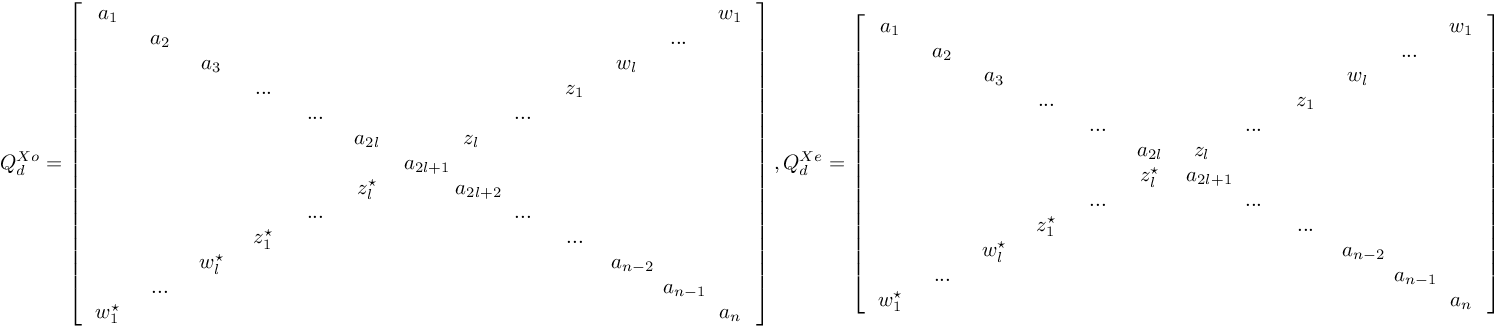}
\end{equation}
where empty cells represent zeros, $n = d^2$, and for $a_i$, $z_i$, $w_i$ we give additional assumptions:
\begin{equation}
	\sum_{i=1}^{n} a_i = 1, a_i \geq 0 , \;\;
	\;
	\prod_{i=1}^{ l }  |z_i|  \leq \sqrt{ \prod_{j=l+1}^{n-l} a_j  },
	\;
	\prod_{i=1}^{ l }  |w_i|  \leq \sqrt{ \left( \prod_{j=1}^{l} a_j \right) \left( \prod_{j=n-l+1}^{n} a_j \right) }  .
\end{equation}
where $l=\floor{\frac{d^2}{4}}$.
 
Fig.~\ref{lbl:fig:X:HY:entropy} depicts the numerical experiment for the state $X$ with dimensions 2, 3, 4, 5 where we check the triangle inequality (HY6) given by Eq.~\ref{lbl:eq:HYsr:triangle:inequality}. This experiment may be also carried out in a parallel environment because each state is checked independently, so it might be realised in the same time. It should be noticed, that the experiment's implementation, thanks to the Python's environment and auxiliary library EntDetector \cite{EntDetector2021}, needs only few lines of code, e.g.  
\begin{verbatim}
	m = create_x_state( d ) 
	pt0 = ed.PT(m, [d,d], 1); pt1 = ed.PT(m, [d,d], 0)	
	ent_m =  HY_by_d( m, s, r )
	ent_pt0 = HY_by_d( pt0, s, r ); ent_pt1 = HY_by_d( pt1, s, r );
\end{verbatim}
In the first line, we calculate state $X$ with $d$ as the dimension of subsystems A and B. In the next line, we compute partial traces (with \verb"ed.PT( )" function). Finally, the entropy is obtained according to Eq.~\ref{lbl:eq:HY:v1}.

\begin{figure}
\begin{tabular}{ccc}
& & \vspace{-0.75cm}\multirow{2}{*}{\includegraphics[width=0.45\textwidth]{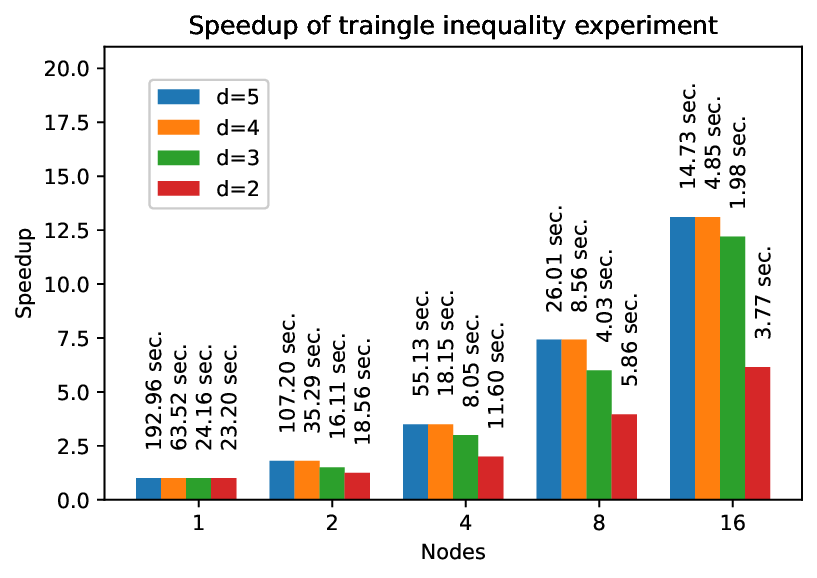}}  \\
\includegraphics[width=0.25\textwidth]{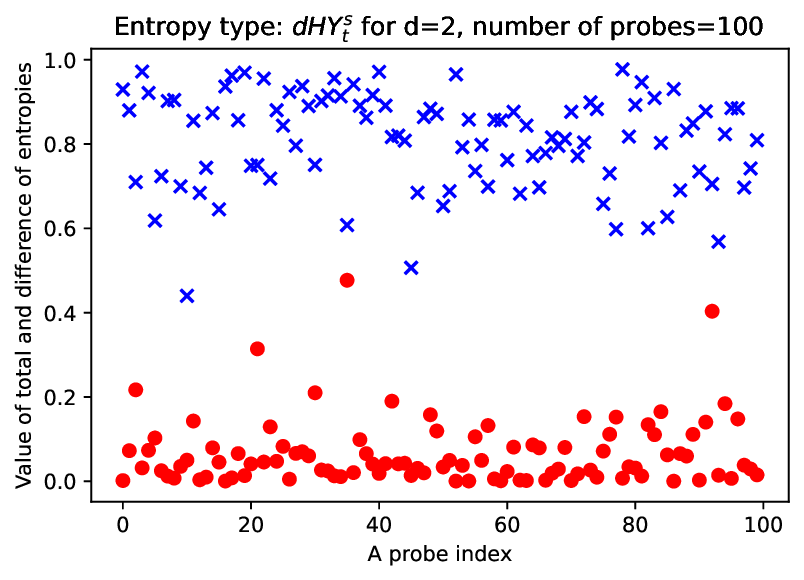} & \includegraphics[width=0.25\textwidth]{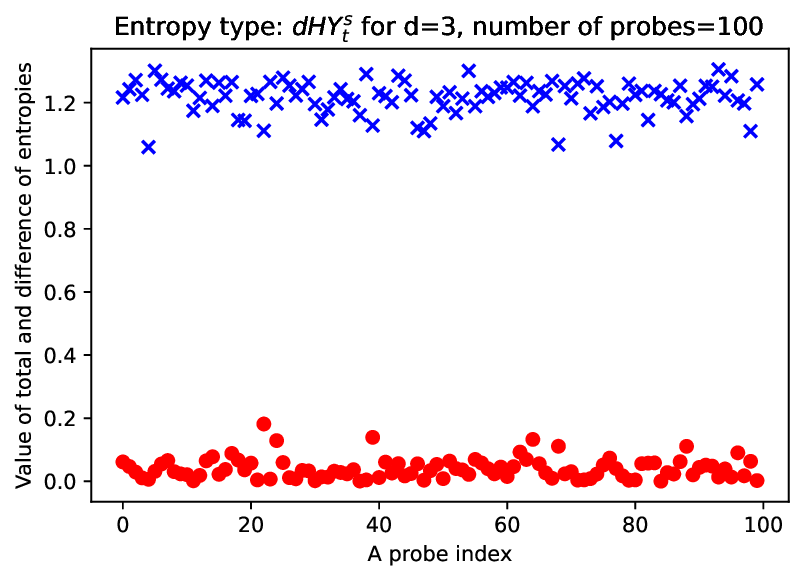} &  \\
\includegraphics[width=0.25\textwidth]{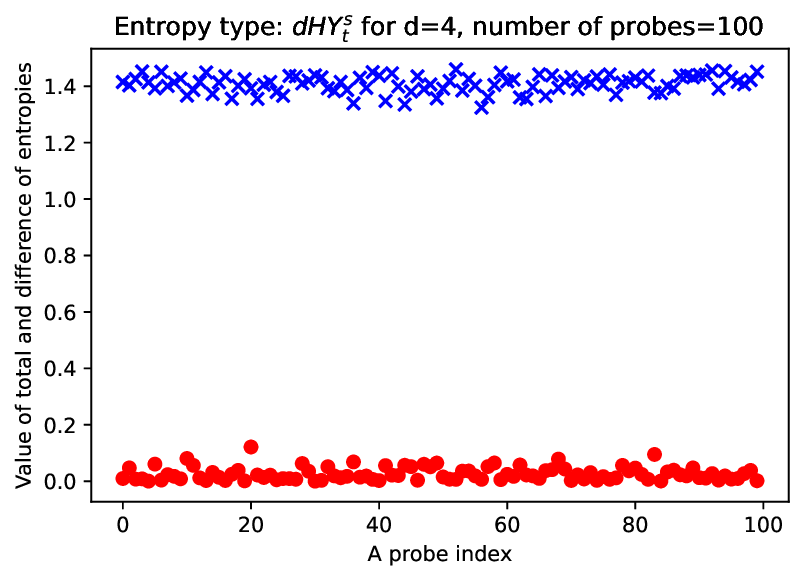} & \includegraphics[width=0.25\textwidth]{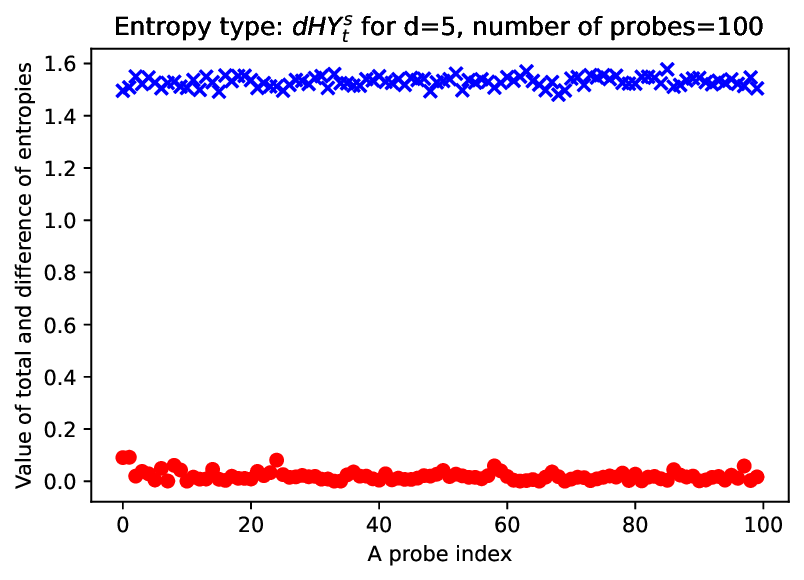} & 
\end{tabular}
\caption{Results of the numerical experiment in which the triangle inequality for unified entropy is verified on the state $X$. The charts represent values for one hundred exemplary states $X$ with, respectively, dimensions $d=2,3,4,5$ (400 samples). Crosses indicate the entropy of the whole system $X$, dots mark the absolute value for subsystems A and B of the state $X$. The utilised entropy is $HY^{s}_r$ (s=0.5, r=2). The chart on the right shows the acceleration gained when the experiment is carried out under WSL environment for Windows 11, with AMD Ryzen 9 7950X processor
}
\label{lbl:fig:X:HY:entropy}
\end{figure}

\subsection{The infinite bipartite case}

The third example that we would like to present is calculating entropy value for the two–mode squeezed Gaussian states given as:
\begin{equation}
\mket{\psi_r} =  \frac{1}{\cosh(r)} \sum_{N=0}^{\infty} \tanh^N(r) \mket{N}_A \otimes \mket{N}_B,
\end{equation}
where $r > 0$ is the squeezing parameter and $\mket{N}$ stands for the "N-particles state". It is stated in \cite{Giedke2003} that the formula for the entropy for this state is: 
\begin{equation}
E_G(\mket{\psi_r}) = \cosh^2(r) \log(\cosh^2(r)) - \sinh^2(r) \log( \sinh^2(r) ).
\label{lbl:eq:egr:formula}
\end{equation}
The value of entropy diverge to infinity according to increasing value of $r$, however even for relatively small values $r>17$ the above equation may generate an overflow error for computations carried out on double-precision numbers. However, the technique presented in this work together with the procedure \verb"fredholm_det", presented at the beginning of Sec.~\ref{lbl:numerical:examples}, correctly approximate the value of $E_G(\mket{\psi_r})$. To perform calculations, we need the kernel function which, in this case, is: 
\begin{equation}
K(x_i, x_j) = \frac{ \tanh(x_i + x_j) }{ \cosh( x_i - x_j ) },
\end{equation}
where $x_i$, $x_j$ are arguments of the quadrature utilised in the procedure \verb"fred"\-\verb"holm_det".
The technique based on the Fredholm determinants allows approximating entropy values with better numerical stability. The results are depict in Fig.~\ref{lbl:fig:BGS:entropy} where values of $E_G(\mket{\psi_r})$ are calculated for two exemplary ranges of parameter $r$.   

\begin{figure}
\begin{center}
\begin{tabular}{cc}
(a) & (b) \\
\includegraphics[width=0.40\textwidth]{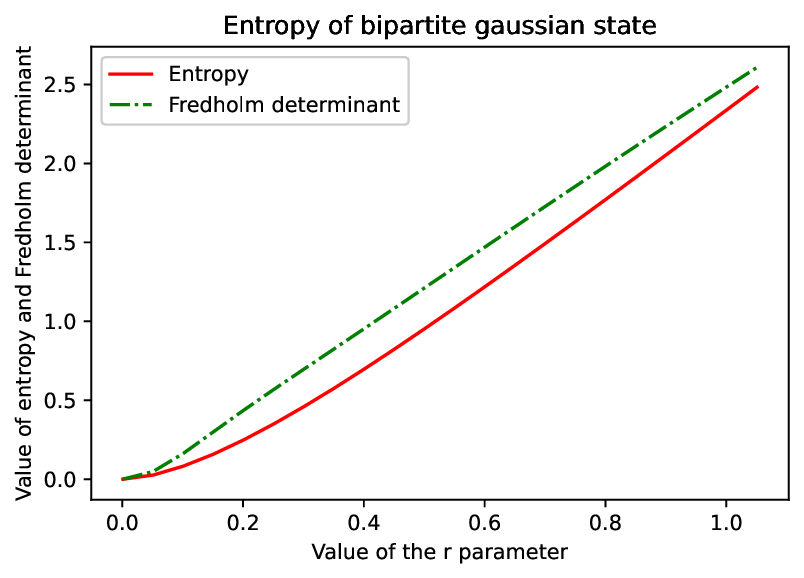} & \includegraphics[width=0.40\textwidth]{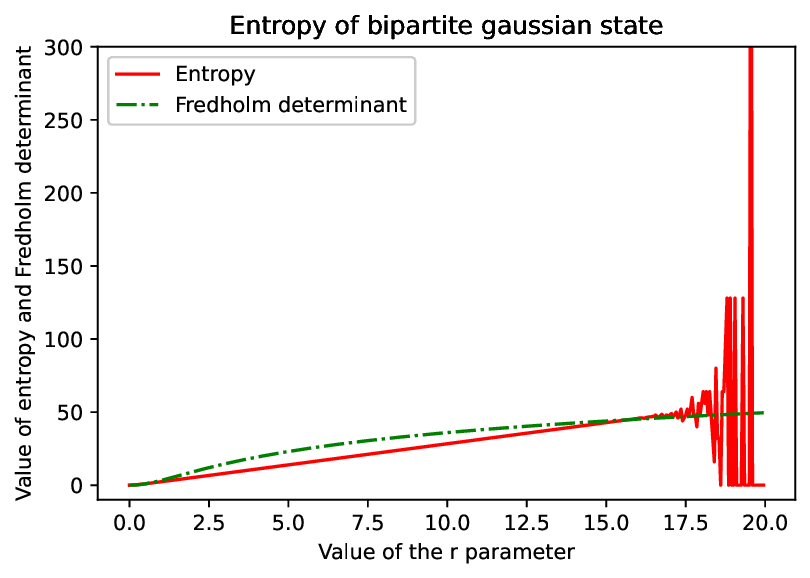} 
\end{tabular}
\end{center}
\caption{Entropy value for a Gaussian bipartite state (the solid line is calculated according to the analytical equation and the dash-dot line with the use of the Fredholm determinants). Chart (a) shows the entropy values for a small range of $r \in (0,1)$ and chart (b) for $r \in (0,20]$. We can observe unstable behaviour of the formula Eq.~\ref{lbl:eq:egr:formula} for $r \approx 20$, where overflow floating-point error appears due to nature of cosh and sinh functions}
\label{lbl:fig:BGS:entropy}

\end{figure}

\section{Conclusions} \label{lbl:conclusions}

In this article, we have shown that the Fredholm determinants may be successfully applied to calculate the entropy values for finite and infinite cases. This second case is especially important because for states described by the Hilbert space, the values of von Neumann entropy might be infinite. The renormalisation process allows calculating the approximate entropy value in many cases without any numerical overflow problems.  

There are still opened issues, e.g. as in Remark~\ref{lbl:remark:HY:properties:question} -- which properties of the renormalised entropy remain true when the parameter $r \in (0,1)$. Further works should also focus on numerical procedures (e.g. utilisation of other quadratures to better estimate the entropy values for finite and infinite cases). 

\begin{credits}
\subsubsection{\ackname}
This work was financed by Military University of Technology under research project UGB 701/2024 and by a subsidy for research projects in Technical Computer Science and Telecommunication discipline in University of Zielona Góra for year 2024.

\subsubsection{\discintname}
All authors declare that they have no conflicts of interest.
\end{credits}
%
%
%
%

\end{document}